\documentclass[letterpaper, 10 pt, conference]{ieeeconf}  % Comment this line out if you need a4paper

\IEEEoverridecommandlockouts                              % This command is only needed if 
\usepackage{graphics} % for pdf, bitmapped graphics files
\usepackage{epsfig} % for postscript graphics files
\usepackage{mathptmx} % assumes new font selection scheme installed
\usepackage{times} % assumes new font selection scheme installed
\usepackage{amsmath} % assumes amsmath package installed

\usepackage{amsthm}
\usepackage{amssymb}  % assumes amsmath package installed
\usepackage[T1]{fontenc}
\usepackage{arydshln}
\usepackage{subfigure} %??
\usepackage{cite}
\usepackage{textcomp}
\usepackage{stfloats}
\usepackage{hyperref}
\hypersetup{hypertex=true,
	colorlinks=true,
	linkcolor=black,
	anchorcolor=black,
	citecolor=black}
\usepackage{algorithm}  
\usepackage{algpseudocode}  
\usepackage{algorithmicx}
\title{\LARGE \bf
Event-Triggered Optimal Formation Tracking Control Using Reinforcement Learning for Large-Scale UAV Systems
}

\author{Ziwei Yan$^{1}$, Liang Han$^{1^*}$, Xiaoduo Li$^{1}$, Jinjie Li$^{2}$ and Zhang Ren$^{2}$% <-this % stops a space
	\thanks{This work was supported in part by the Science and Technology Innovation 2030 Key Project of "New Generation Artificial Intelligence" under Grant 2018AAA0102305, and in part by the National Natural Science Foundation of China under Grant 61803014 and Grant 61873011.}
\thanks{*Corresponding author: L. Han,
	{\tt\small liang\_han@buaa.edu.cn.}}% <-this % stops a space
\thanks{$^{1}$Z. Yan, L. Han and X. Li are with the Sino-French Engineer School, Beihang University, Beijing 100191, China.
        {\tt\small yanziwei@buaa.edu.cn; liang\_han@buaa.edu.cn; lixiaoduo@sjtu.edu.cn.}}%
\thanks{$^{2}$J. Li and Z. Ren are with the School of Automation Science and Electrical Engineering, Beihang University, Beijing, 100191, China. {\tt\small lijinjie@buaa.edu.cn; renzhang@buaa.edu.cn.}}
}

\begin{document}

\maketitle
\thispagestyle{empty}
\pagestyle{empty}
\newtheorem{definition}{Definition}
\newtheorem{assumption}{Assumption}
\newtheorem{theorem}{Theorem}
\newtheorem{lemma}{Lemma}
\newtheorem{remark}{Remark}

%%%%%%%%%%%%%%%%%%%%%%%%%%%%%%%%%%%%%%%%%%%%%%%%%%%%%%%%%%%%%%%%%%%%%%%%%%%%%%%%
\begin{abstract}
Large-scale UAV switching formation tracking control has been widely applied in many fields such as search and rescue, cooperative transportation, and UAV light shows. In order to optimize the control performance and reduce the computational burden of the system, this study proposes an event-triggered optimal formation tracking controller for discrete-time large-scale UAV systems (UASs). And an optimal decision - optimal control framework is completed by introducing the Hungarian algorithm and actor-critic neural networks (NNs) implementation. Finally, a large-scale mixed reality experimental platform is built to verify the effectiveness of the proposed algorithm, which includes large-scale virtual UAV nodes and limited physical UAV nodes. This compensates for the limitations of the experimental field and equipment in real-world scenario, ensures the experimental safety, significantly reduces the experimental cost, and is suitable for realizing large-scale UAV formation light shows.

\end{abstract}

%%%%%%%%%%%%%%%%%%%%%%%%%%%%%%%%%%%%%%%%%%%%%%%%%%%%%%%%%%%%%%%%%%%%%%%%%%%%%%%%
\section{Introduction}
In recent years, enormous studies on cooperative control for multi-UAV system are overwhelming. Among them, multi-UAV formation tracking control is a common and fundamental problem, and the formation scale has changed from a few UAV groups previously to large-scale formations nowadays. With the exponential growth of formation scale, the system performance is more rigorously required. Hence, researchers proposed optimal control that optimize system performance and formation cost to reduce the system burden.

Optimal control is an effective strategy to balance the system performance and computational consumption. %The common methods include the minimum principle \cite{c1,c2}, and the dynamic programming \cite{c3,c4}.
For discrete-time systems, the optimal control can be achieved by solving the Bellman equation \cite{c1,c2,c3,c4}. However, since the analytic solution of the Bellman nonlinear equation is difficult to be derived directly, all studies above consider the linear systems. In order to approximate optimal controllers for nonlinear systems, reinforcement learning (RL) has been introduced in \cite{c5,c6,c7,c8,c81,c82}, where the actor-critic neural network (NN) is the commonly used RL architecture. The critic NN gives feedback to optimize the actions by evaluating the system performance index, while the actor NN issues optimized control commands to improve the system behaviors. 
\begin{figure}[!htbp]
	\centering
	\includegraphics[width=\columnwidth]{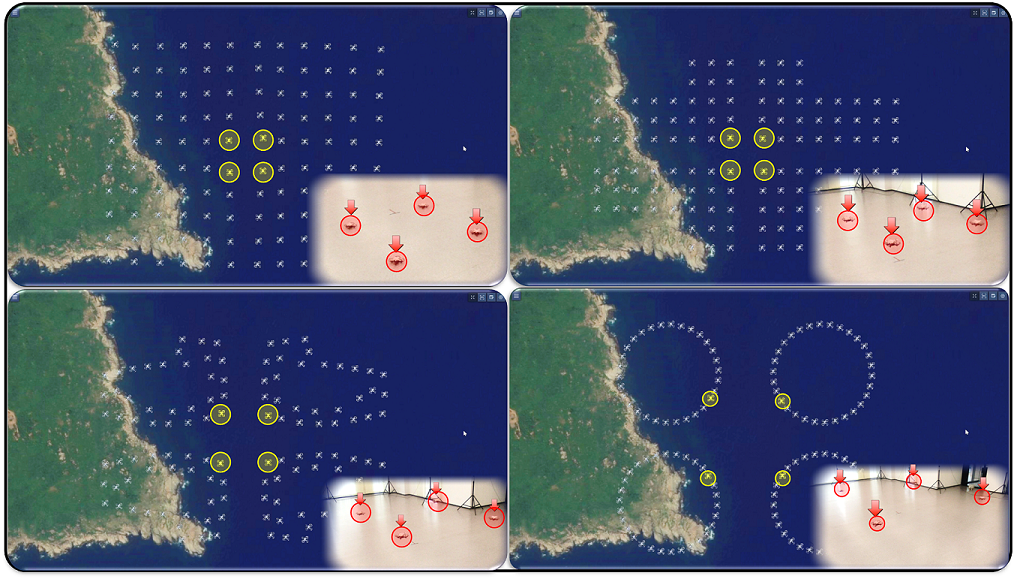}
	\caption{Hundreds of virtual UAVs and four physical UAVs in a large-scale mixed reality experimental platform to realize a UAV formation light show. \label{Fig0}}
	\vspace{-2em} 
\end{figure} 
%\vspace{-0.8cm} 

Last few years, RL architectures have been widely used in optimal cooperative control. In \cite{c9}, a learning-based adaptive dynamic programming algorithm is used to solve the formation tracking problem for multi-UAV systems with model uncertainty in obstacle environments. In \cite{c10}, for the multi-agent tracking control with actuator faults, an adaptive optimal fault-tolerant  tracking controller is designed by combining the RL algorithm and the backstepping method. %In \cite{c11}, a model-free optimal fault-tolerant containment controller using RL method is proposed to solve the containment problem for MASs with time-varying actuator faults. 
However, the studies above consider the time-triggered mechanism, which will continuously or periodically update the controller and NNs. This will inevitably cause a huge computational consumption when the formation size is large.

This study introduces the event-triggered mechanism where the UAV updates the controller and actor NN only when the specific condition is triggered, which allows the UAV system (UAS) to update control commands on demand, aperiodically rather than continuously and indiscriminately. At the same time, it means that the event-triggered mechanism can effectively reduce the system computational burden compared to time-triggered cases. For the RL-based optimal cooperative problems, a number of studies have introduced the event-triggered mechanism as an improvement, whether for continuous-time systems \cite{c12,c13,c14,c14_1,c14_2,c14_3} and discrete-time systems \cite{c15,c16,c17,c17_1,c17_2,c17_3}. However, these studies only implemente numerical simulations without verifying the algorithm effectiveness for engineering applications. 

In order to verify the efficiency of the proposed algorithm, a large-scale mixed reality experimental platform is established in this study. A large-scale UAV light show is realized using a limited number of physical UAVs and hundreds of virtual UAVs (Fig. \ref{Fig0}). In view of the main limitations of large-scale formation show, such as high requirement for experimental field, expensive experimental equipment, and poor safety factor, the established experimental platform has the advantages of low experimental cost and high safety factor. And the experimental results show that this platform only requires one room and four UAVs to achieve the desired effect of hundreds of UAVs in the large airspace.

This paper proposes an event-triggered optimal formation tracking controller for discrete-time large-scale UASs, and an online actor-critic NN structure is established to implement the control. The main contributions of this study include:
\begin{enumerate}
	\item Considering the main difficulties involved in the decision and control layers for large-scale switching formations, an algorithm framework integrating optimal formation assignment and optimal control is established.
	%\item Different from previous studies \cite{c81,c82,c9} on time-triggered optimal formation tracking control, this study introduces the event-triggered mechanism to optimize the system computational efficiency.
	%\item Different from previous studies \cite{c81,c82,c9} on optimal formation tracking control, where the controllers are all continuously updated under the time-triggered mechanism. This study introduces the event-triggered mechanism to reduce the controller update frequency and optimize the system computational efficiency.
	\item The event-triggered online actor-critic NN is proposed to approximate the local performance index and learn controller. Compared with \cite{c5,c6,c7,c8,c81,c82,c9,c10}, actor NN weights and controller only update at triggering instants.
	%\item RL-based methods are widely used to solve the Bellman equation involved in optimal control. This paper proposes event-triggered online actor-critic NNs to approximate the system performance index and learn the controller. In contrast to \cite{c5,c6,c7,c8,c81,c82,c9,c10,c11}, the actor network weights only updates at the triggering instants.
	\item This study highlights the engineering application of the proposed algorithm, establishes a mixed reality experimental platform for large-scale UAV formation, and verifies the proposed algorithm on this platform.
	%\item Compared with studies using numerical simulations to verify the effectiveness of proposed controllers,  this study highlights the engineering application properties of the proposed algorithm, establishes a virtual-real experimental platform for large-scale UAV formation, and verifies the proposed algorithm on this platform. 
\end{enumerate}
\begin{figure*}[!pt]% h here; t top; b bottom; p page.
	\centering
	\includegraphics[width=\linewidth,scale=1.00]{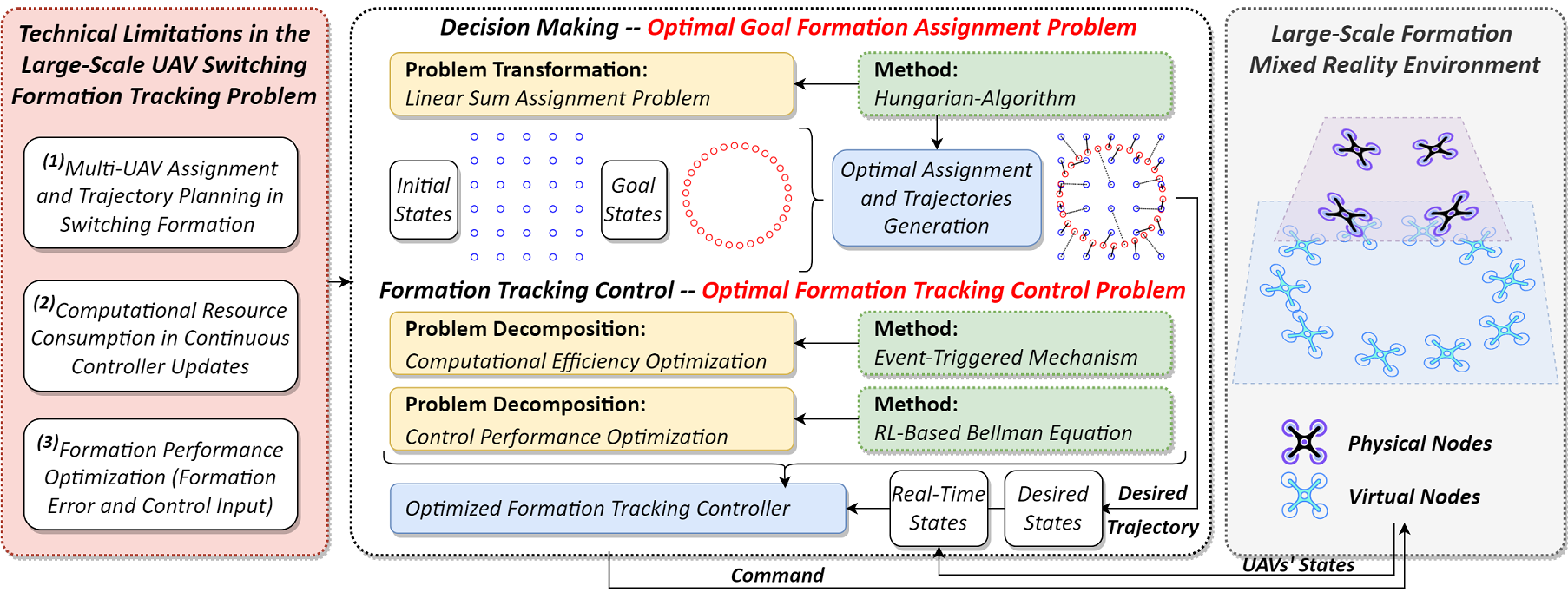}
	\caption{The algorithm overview of this study. To solve the large-scale UAV switching formation tracking problem, for the decision layer, the Hungarian algorithm is used to solve the goal formation assignment problem. For the control layer, the event-triggered mechanism and the RL-based Bellman equation are used to optimize the computational efficiency and the control performance of the system, respectively.}	
	\label{Fig1}
	\vspace{-0.5em} 
\end{figure*}
\iffalse
\subsection{Organization}
This paper is organized as follows. In Section \ref{sec2}, an algorithm overview proposed in this study is presented. In Section \ref{sec3}, the basic theory and problem description are listed. In Section \ref{sec4}, the online actor-critic NNs are proposed to implement the designed event-triggered optimal controller. In Section \ref{sec5}, the mixed reality experimental platform is established and the experimental results are analyzed. And this study is concluded in Section \ref{sec6}.
\fi
\section{Algorithm Overview}\label{sec2}
 As shown in Fig. \ref{Fig1}, an algorithm framework is designed to solve the large-scale UAV optimal formation tracking problem. And the algorithms for optimal assignment and optimal control are summarized in Algorithm \ref{alg:1} and Algorithm \ref{alg:2}.
%\begin{enumerate}
%	\item Multi-UAV formation task assignment and trajectory planning in large-scale switching formations;
%	\item Massive computational resource cost for continuous controller updates in large-scale formation control;
%	\item Formation performance optimization, including formation error optimization and control input optimization.
%\end{enumerate}  
\begin{algorithm}[!hb] 
	\caption{Decision Layer - Optimal Assignment}
	\label{alg:1}
	
	{\bf Input:} 1) Initial position set $\vec{\mathbf{G}}$; 2) Expected formation is $\vec{\mathbf{S}}$; \\ \hspace*{0.43in}3) UAV number $N$.
	
	{\bf Output:} 1) Target position set $\vec{\mathbf{F}}$; 2) Optimal assignment
	\\ \hspace*{0.48in} list $\chi^*$; 3) Desired switching trajectory $h(k)$.
	\begin{algorithmic}[1] % remove [1] if you do not need row number in your algorithm
		\State [$\vec{\mathbf{F}}$, $\chi^*$, $h(k)$]=Optimal\_Assignment($\vec{\mathbf{G}}$, $\vec{\mathbf{S}}$, $N$).
	\end{algorithmic}
\end{algorithm}
\begin{algorithm}[!ht]
	\caption{Control Layer - Optimal Control}
	\label{alg:2}
	
	{\bf Input:} 1) Desired trajectory $h(k)$ 2) Real-time states $x(k)$.
	
	{\bf Output:} 1) Optimized controller $u(k)$; 2) Triggering instants\\ \hspace*{0.48in}  sequence $\{k_l^i\}_{l\in\mathbb{Z}^+}$.

	\begin{algorithmic}[1] % remove [1] if you do not need row number in your algorithm
		\State Initialize the actor-critic NN;
		\State Calculate triggering condition $f_i(k)$ with $h_i(k)$ and $x_i(k)$;
		\If{$f_i(k)>0$}
		\State Using actor NN to learn the controller $u_i(k)$;
		\State Calculate and update the actor NN;
		\State The corresponding instant $k$ is recorded in $\{k_l^i\}_{l\in\mathbb{Z}^+}$. 
		\Else
		\State $u_i(k+1)=u_i(k)$;
		\State The actor NN remains the same as before.
		\EndIf
		\State Calculate and update the critic NN;
		\State Approximate the system performance index.
	\end{algorithmic}
\end{algorithm}
%Firstly, for the optimal goal formation assignment problem at the decision layer, it can be transformed into the linear sum assignment problem (LSAP) in \cite{c18}, then using the Hungarian algorithm to solve it and generate the expected trajectory, which is passed to the control layer as the desired formation. Then, to optimize the large-scale UAV formation tracking control, the event-triggered mechanism is used to optimize the controller update frequency, % to reduce the computational consumption, 
%and the Bellman equation is solved using RL method to achieve the formation error optimization and control input optimization. Finally, combined with the above optimization algorithms, the optimized formation tracking controller is proposed.

\section{Preliminaries} \label{sec3}

\subsection{Graph Theory}
For a large-scale UAS with a virtual leader and $N$ followers, the communication network can be described as a graph $\mathcal{G}=\{\mathcal{A},\mathcal{E},\mathcal{W}\}$, where $\mathcal{A}=\{a_1,\dots,a_N\}$ denotes the node set, $\mathcal{E}$ is the edge set and $e_{ij}=(a_i,a_j)\in\mathcal{E}$, and $\mathcal{W}$ is the adjacency matrix with elements $w_{ij}$, for $e_{ij}\in\mathcal{E},~i\neq j$, $w_{ij}=1$, otherwise $w_{ij}=0$. Define $N_i$ as the neighbor set of UAV $a_i$. Then, the in-degree matrix $\mathcal{D}=\mathrm{diag}_{i\in[1,N]}\{\sum_{j\in N_i}w_{ij}\}$. Laplacian matrix is $L=\mathcal{D}-\mathcal{W}$. Further, a diagonal matrix $B=\mathrm{diag}_{i\in[1,N]}\{b_i\}$ describes the communication between followers and the virtual leader. Assume that the leader can communicate with at least one follower, and there exists a spanning tree with the leader as the root in the directed group $\hat{\mathcal{G}}=\{\hat{\mathcal{A}},\hat{\mathcal{E}}\}$ with $\hat{\mathcal{A}}=\{a_0,\dots,a_N\}$ and $\hat{\mathcal{E}}\in\hat{\mathcal{A}}\times\hat{\mathcal{A}}$. %Then, it can be concluded that all eigenvalues of the matrix $L+B$ are positive.
\begin{assumption}
	The graph $\mathcal{G}$ is connected and the directed graph $\hat{\mathcal{G}}$ contains a spanning tree.
\end{assumption}

\subsection{Radial Basis Function (RBF) NNs}
An unknown nonlinear function $\zeta(\varOmega)$ can be approximated by a RBF NN,	$\zeta(\varOmega)=W^{*\mathrm{T}}\psi(\varOmega)$, where $W^*\in\mathbb{R}^s$ is a weight vector, $\psi(\varOmega)\in\mathbb{R}^s$ is a basis function vector with $\psi(\varOmega)=\left[\psi_1(\varOmega),\dots,\psi_s(\varOmega)\right]^\mathrm{T}$, and $s\in\mathbb{Z^+}$ represents the number of NN nodes in the hidden layer. In this paper, $\psi_i(\varOmega)$ is chosen as the Gaussian Basis Function
\begin{equation}
	\psi_i(\varOmega)=\exp\left(-(\varOmega-\nu_i)^\mathrm{T}(\varOmega-\nu_i)/2\iota_i^2 \right),~i=1,\dots,s.
\end{equation}
where $\nu_i$ is the center of the Gaussian kernel function and $\iota_i$ is the width parameter of the function. For the approximation of the RBF NN, it is always possible to find an ideal weight $W^*=\arg_{\min W}\{\sup_{\varOmega}||\zeta(\varOmega)-W^\mathrm{T}\psi(\varOmega)||\}$. As a result, the unknown function $\zeta(\varOmega)$ can be written as
\begin{equation}
	\zeta(\varOmega)=W^{*\mathrm{T}}\psi(\varOmega)+\sigma(\varOmega).
\end{equation}
where $\sigma(\varOmega)$ is the approximation error with $|\sigma(\varOmega)|\leq\sigma_T$, and $\sigma_T$ is a constant.

\subsection{Problem Description}
Consider a large-scale UAS with $N$ followers and a virtual leader, each UAV $a_i$ $(i\in\{0,\dots,N\})$ is modeled as:
\begin{equation}\label{multi-UAV system}
	\begin{cases}
		p_i(k+1)=p_i(k)+v_i(k)T\\
		v_i(k+1)=v_i(k)+u_i(k)T+f_i(p_i(k),v_i(k)).
	\end{cases}
\end{equation}
where $p_i(k)\in\mathbb{R}^m$ and $v_i(k)\in\mathbb{R}^m$ are the position and velocity of the UAV $a_i$, $u_i(k)\in\mathbb{R}^m$ is the control input with $m$ the movement dimension of UAV, and $f_i(\cdot)$ is the unknown bounded nonlinear function. For the virtual leader $a_0$, $u_0(k)$ is known and $f_0(\cdot)=0$.
%where $p_0(k)\in\mathbb{R}^m$ and $v_0(k)\in\mathbb{R}^m$ are the position and velocity of the leader UAV $a_0$, and $u_0(k)\in\mathbb{R}^m$ is the control input.

Then, for each follower UAV $a_i$, the position tracking error $\xi_{pi}\in\mathbb{R}^m$ and velocity tracking error $\xi_{vi}\in\mathbb{R}^m$ are defined as
\begin{equation}
	\begin{cases}
		\xi_{pi}(k)=p_i(k)-p_0(k)-\eta_{pi}(k)\\
		\xi_{vi}(k)=v_i(k)-v_0(k)-\eta_{vi}(k).
	\end{cases}
\end{equation}
where $\eta_i(k)=[\eta_{pi}^\mathrm{T}(k),\eta_{vi}^\mathrm{T}(k)]^\mathrm{T}\in\mathbb{R}^{2m}$ is the desired state difference between the leader $a_0$ and the follower $a_i$. Then, define the disagreement error for UAV $a_i$ as
\begin{equation}\label{FTE}
	\begin{cases}
		\begin{aligned}
			\varepsilon_{pi}(k)&=\sum_{j\in N_i}w_{ij}(\xi_{pi}(k)-\xi_{pj}(k))+b_i\xi_{pi}(k)\\
			\varepsilon_{vi}(k)&=\sum_{j\in N_i}w_{ij}(\xi_{vi}(k)-\xi_{vj}(k))+b_i\xi_{vi}(k).
		\end{aligned}
	\end{cases}
\end{equation}
Let $\varepsilon(k)=[\varepsilon_1^\mathrm{T}(k),\dots,\varepsilon_N^\mathrm{T}(k)]^\mathrm{T}$ with $\varepsilon_i(k)=[\varepsilon_{pi}^\mathrm{T}(k),\varepsilon_{vi}^\mathrm{T}(k)]^\mathrm{T}$.
\iffalse
and the disagreement error dynamics is written as
\begin{equation}\label{error dynamics}
	\begin{cases}
		\begin{aligned}
			\varepsilon_{pi}(k+1)&=\varepsilon_{pi}(k)+\varepsilon_{vi}(k)T\\
			\varepsilon_{vi}(k+1)&=\varepsilon_{vi}(k)+\gamma_i(u_i(k)T+f_i)\\
			&~~~-\left\lbrace \sum_{j\in N_i}w_{ij}(u_j(k)T+f_j)+b_iu_0(k)T\right\rbrace 
		\end{aligned}
	\end{cases}
\end{equation} 
where $\gamma_i=\sum_{j\in N_i}w_{ij}+b_i$. Let $\varepsilon(k)=[\varepsilon_1^\mathrm{T}(k),\dots,\varepsilon_N^\mathrm{T}(k)]^\mathrm{T}$ with $\varepsilon_i(k)=[\varepsilon_{pi}^\mathrm{T}(k),\varepsilon_{vi}^\mathrm{T}(k)]^\mathrm{T}$.
\fi
Define the local performance index function (PIF) as
\begin{equation}\label{PIF}
	V_i(\varepsilon_i(k),u_i(k))=\sum_{l-k}^{\infty}\beta^{l-k}R_i(\varepsilon_i(k),u_i(\varepsilon_i)),
\end{equation}
where $R_i(\varepsilon_i(k),u_i(\varepsilon_i))=\varepsilon_i^\mathrm{T}(k)\varepsilon_i(k)+u_i^\mathrm{T}(\varepsilon_i)u_i(\varepsilon_i)$ is the utility function.% For any admissible control variable $u(k)\in\mathrm{U}$, the local PIF can be written as the following Bellman equation:
%\begin{equation}
%	V_i(\varepsilon_i(k),k)=\min\{R_i(\varepsilon_i(k),u_i(\varepsilon_i))+\beta V_i(\varepsilon_i(k+1),k+1)\}.
%\end{equation}

\iffalse
Then, the optimal formation tracking controller is defined as 
\begin{equation}
	u^*_i=\arg\min_{u_i(\varepsilon_i)}\{R_i(\varepsilon_i(k),u_i(\varepsilon_i))+\beta V^*_i(\varepsilon_i(k+1))\}
\end{equation}

\begin{definition}
	The formation tracking error for the system is $\xi(k)=[\xi_{1}^\mathrm{T}(k),\dots,\xi_{N}^\mathrm{T}(k)]^\mathrm{T}$ with $\xi_{i}(k)=[\xi_{pi}^\mathrm{T}(k),\xi_{vi}^\mathrm{T}(k)]^\mathrm{T}$. If for any follower UAV $a_i$ $(i\in\{1,\dots,N\})$, $\lim_{t\rightarrow\infty}||\xi_{i}(k)||=0$, the desired formation tracking control is achieved.
\end{definition}
\fi

\section{Optimal Assignment and Optimal Control}\label{sec4}
In this section, the algorithm overview shown in Fig. \ref{Fig1} is expanded and detailed separately in three parts: optimal formation assignment, event-triggered optimal controller design, and online actor-critic NNs implementation. 
\begin{figure}[!ht]
	\centering
	\includegraphics[width=\columnwidth]{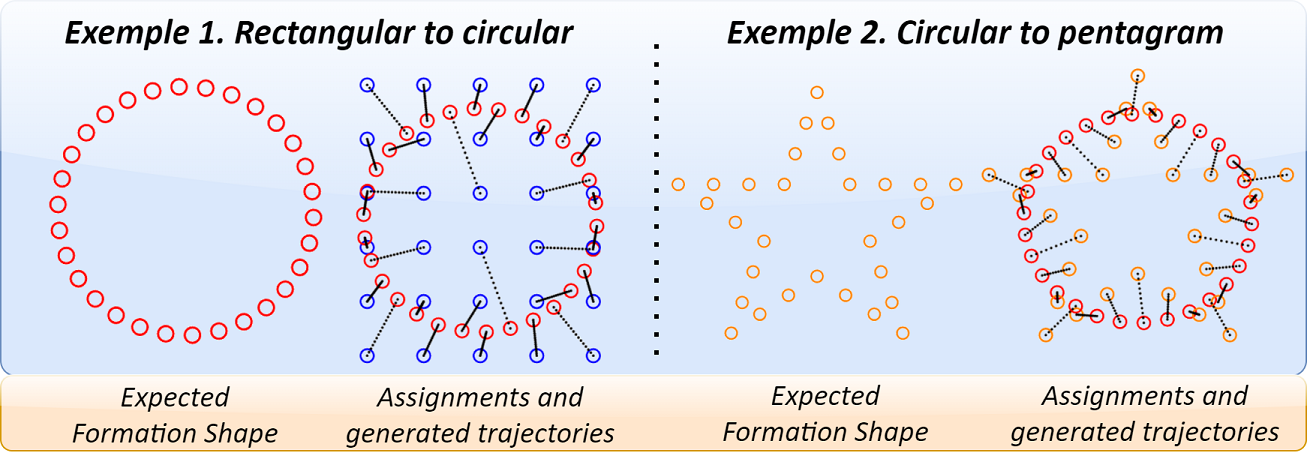}
	\caption{Two examples of switching formations, with the expected formation shape on the left and the assignments and generated trajectories on the right. \label{Fig7}}
	 \vspace{-2em}  % ????????
\end{figure}
\subsection{Optimal Assignment for Switching Formations}
For a system with $N$ UAVs of radius $r$, suppose its initial position set is $\vec{\mathbf{G}}=\{\vec{\mathbf{g}}_i^\mathrm{T}\}\in\mathbb{R}^{N\times m}$, the expected formation is $\vec{\mathbf{S}}=\{\vec{\mathbf{s}}_j^\mathrm{T}\}\in\mathbb{R}^{N\times m}$, and the target position set is $\vec{\mathbf{F}}=\{\vec{\mathbf{f}}_j^\mathrm{T}\}\in\mathbb{R}^{N\times m}$ with $\vec{\mathbf{f}}_j=\rho\vec{\mathbf{s}}_j+\vec{d}$, where $\rho\in\mathbb{R}^+$ is the scale factor and $\vec{d}$ is the translation vector. 
 \begin{lemma}\cite{c18}
 	The optimal scale factor $\rho^*$ and translation vector $\vec{d}^*$ can be calculated as
 	\begin{equation}
 		\begin{aligned}
 			\rho^*&=\dfrac{\sum_{i=1}^{N}\vec{\mathbf{g}}_i(\sum_{j=1}^{N}\vec{\mathbf{s}}_j)^\mathrm{T}+N\kappa^*}{\sum_{j=1}^{N}\vec{\mathbf{s}}_j(\sum_{j=1}^{N}\vec{\mathbf{s}}_j)^\mathrm{T}-N\sum_{j=1}^N\vec{\mathbf{s}}_j\vec{\mathbf{s}}_j^\mathrm{T}},\\
 			\vec{d}^*&=\dfrac{\sum_{i=1}^{N}\vec{\mathbf{g}}_i-\rho^*\sum_{j=1}^{N}\vec{\mathbf{s}}_j}{N}.
 		\end{aligned}
 	\end{equation} 
 \end{lemma}
\begin{algorithm}[!hb] 
	\caption{\cite{c18} Pseudo-Cost Assignment Algorithm}
	\label{alg:3}
	
	{\bf Input:} 1) Initial position set $\vec{\mathbf{G}}$; 2) Expected formation is $\vec{\mathbf{S}}$.
	
	{\bf Output:} 1) Optimal assignment list $\chi^*\in\mathbb{Z}^N$; 2) Minimum\\ \hspace*{0.48in} cost matrix $\kappa^*\in\mathbb{R}^{N\times N}$. 
	\begin{algorithmic}[1] % remove [1] if you do not need row number in your algorithm
		\State $\kappa_{ij}=-\vec{\mathbf{g}}_i\vec{\mathbf{s}}_j^\mathrm{T}.$
		\State $(\chi^*,\kappa^*)=\text{Hungarian-LSAP}(\kappa).$
	\end{algorithmic}
\end{algorithm}
\iffalse
\begin{lemma}\cite{c18}
	The cost matrix is set as $\mathcal{C}_{ij}=||\vec{\mathbf{g}}_i-\vec{\mathbf{f}}_j||_2^2$, which can be reduced to a pseudo-cost matrix $\kappa=-\vec{\mathbf{G}}_i\vec{\mathbf{S}}_j^\mathrm{T}$ with derivations. 
\end{lemma}

\begin{lemma}\cite{c18}
	Using the Hungarian-LSAP (H-LSAP) algorithm, the optimal assignment list $\chi^*\in\mathbb{Z}^N$ and the minimum cost matrix $\kappa^*\in\mathbb{R}^{N\times N}$ can be calculated
	\begin{equation}
		(\chi^*,\kappa^*)=\text{H-LSAP}(\kappa).
	\end{equation}
\end{lemma}
\fi

%\begin{equation}
%	\begin{aligned}
%		\rho^*&=\dfrac{\sum_{i=1}^{N}\vec{\mathbf{g}}_i(\sum_{j=1}^{N}\vec{\mathbf{s}}_j)^\mathrm{T}+N\kappa^*}{\sum_{j=1}^{N}\vec{\mathbf{s}}_j(\sum_{j=1}^{N}\vec{\mathbf{s}}_j)^\mathrm{T}-N\sum_{j=1}^N\vec{\mathbf{s}}_j\vec{\mathbf{s}}_j^\mathrm{T}},\\
%		\vec{d}^*&=\dfrac{\sum_{i=1}^{N}\vec{\mathbf{g}}_i-\rho^*\sum_{j=1}^{N}\vec{\mathbf{s}}_j}{N}.
%	\end{aligned}
%\end{equation} 
  Algorithm \ref{alg:3} is introduced to enable UAVs to reach the target positions simultaneously without collision. Define the formation switching duration as $t_s=\max_{i\in[1,N]}||\vec{\mathbf{g}}_i-\vec{\mathbf{f}}_{\chi(i)}||_2/\\v_{\max}$, where $\chi(i)$ is the position index assigned to UAV $a_i$, and $v_{\max}$ is the maximum allowable velocity. Then, the desired switching trajectory for UAV $a_i$ is
\begin{equation}
	h_i(k)=\vec{\mathbf{g}}_i+\dfrac{\vec{\mathbf{f}}_{\chi(i)}-\vec{\mathbf{g}}_i}{t_s}kT,\qquad k\in[0,\dfrac{t_s}{T}].
\end{equation}
And the collision avoidance conditions of UAVs are noted $||\vec{\mathbf{g}}_i-\vec{\mathbf{g}}_j||_2\geq2\sqrt{2}r$ and $||\vec{\mathbf{s}}_i-\vec{\mathbf{s}}_j||_2\geq2\sqrt{2}r$.
\subsection{Event-Triggered Optimal Controller Design}
%For the control layer, in order to reduce the computational resource consumed by updating the controller in the large-scale UAS $(\ref{multi-UAV system})$, the event-triggered optimal formation tracking controller will be proposed.

Define the triggering instants sequence for the UAV $a_i$ as $\{k_l^i\}_{l\in\mathbb{Z}^+}$, which is determined by the following condition,
\begin{equation}\label{triggering condition}
	f_i(k)=||e_l^i(k)||^2-\dfrac{(1-2\kappa^2)}{2\kappa^2}||\varepsilon_i(k)||^2,
\end{equation}
where $\kappa\in(\dfrac{1}{2},\dfrac{\sqrt{2}}{2})$, $\varepsilon_i(k)$ is the current disagreement error and the controller $u_i(k)$ only updates at the triggering instant $k_l^i=\inf\{k>k_{l-1}^i,f_i(k)>0\}$, that is, $\hat{u}_i(k)=u_i(\varepsilon_i(k_l^i))$, where $\varepsilon_i(k_l^i)$ is the triggering error. The zero-order holder (ZOH) is used until next triggering instant occurs, one has
\begin{equation}
	u_i(k)=\hat{u}_i(k)=u_i(k_l^i),\quad k\in[k_l^i,k_l^{i+1}).
\end{equation}
\iffalse
And the event-triggered Bellman equation is defined as
\begin{equation}\label{ET Bellman}
	V_i^*(\varepsilon_i(k))=\min_{u_i(k)}\{R_i(\varepsilon_i(k),u_i(k_l^i))+\beta V_i^*(\varepsilon_i(k+1))\}.
\end{equation}
The event-triggered optimal controller is derived as
\begin{equation}\label{event-triggered controller update}
	u_i^*(k)=\arg\min_{u_i(k_l^i)}\{R_i(\varepsilon_i(k),u_i(k_l^i))+\beta V_i^*(\varepsilon_i(k+1))\}.
\end{equation} 
\fi
\begin{lemma}\cite{c16}\label{assumption3}
	There exists a positive constant $\kappa$ satisfying the inequality as follows,
	\begin{equation}
		||\varepsilon_i(k+1)||=||g_i(\varepsilon_i(k),u_i(k_l^i))||\leq\kappa||\varepsilon_i(k)||+\kappa||e_l^i(k)||
	\end{equation} 
	where $e_l^i(k)=\varepsilon_i(k_l^i)-\varepsilon_i(k),~k\in[k_l^i,k_l^{i+1})$.
And when the event is triggered, $e_l^i(k)=0$.
\end{lemma}
\begin{theorem}
	Considering the large-scale UAS $(\ref{multi-UAV system})$ following the Assumption \ref{assumption3} and the triggering condition (\ref{triggering condition}), the disagreement error is asymptotically stable and formation tracking control for the large-scale UAS $(\ref{multi-UAV system})$ is achieved.
\end{theorem}
\begin{proof}
	$\mathbf{(a)}$ For $\forall k\in[k_{l}^i,k_{l+1}^i)$, the controller $u_i(k)$ remains a constant, $u_i(k)=u_i(k+1)$. Design the Lyapunov function
	\begin{equation}
		\mathcal{V}_{ei}(k)=\varepsilon_i^\mathrm{T}(k)\varepsilon_i(k)+u_i^\mathrm{T}(k)u_i(k).
	\end{equation}
	%Then, the first difference of the $\mathcal{V}_{ei}$ is
	%\begin{equation}
	%	\begin{aligned}
	%		\Delta\mathcal{V}_{ei}(k)%&=\varepsilon_i^\mathrm{T}(k+1)\varepsilon_i(k+1)+u_i^\mathrm{T}(k+1)u_i(k+1)\\
	%		%&~~~-\varepsilon_i^\mathrm{T}(k)\varepsilon_i(k)-u_i^\mathrm{T}(k)u_i(k)\\
	%		%~&=\varepsilon_i^\mathrm{T}(k+1)\varepsilon_i(k+1)-\varepsilon_i^\mathrm{T}(k)\varepsilon_i(k)\\
	%		&=||\varepsilon_i(k+1)||^2-||\varepsilon_i(k)||^2.
	%	\end{aligned}
	%\end{equation}
	With the Lemma $\ref{assumption3}$ and the triggering condition $(\ref{triggering condition})$, $\Delta\mathcal{V}_{ei}(k)=\mathcal{V}_{ei}(k+1)-\mathcal{V}_{ei}(k)$ can be derived as
	\begin{equation}
		\begin{aligned}
			\Delta\mathcal{V}_{ei}(k)%&=||\varepsilon_i(k+1)||^2-||\varepsilon_i(k)||^2\\
			%&\leq (\kappa||\varepsilon_i(k)||+\kappa||e_l^i(k)||)^2-||\varepsilon_i(k)||^2\\
			&\leq 2\kappa^2||\varepsilon_i(k)||^2+2\kappa^2||e_l^i(k)||^2-||\varepsilon_i(k)||^2\leq 0.
		\end{aligned}	
	\end{equation}
	
	$\mathbf{(b)}$ For $k\in\{k_l^i\}_{l\in\mathbb{Z}^+}$, the controller $u_i(k)$ will be updated, and the Lyapunov function is designed as
	\begin{equation}
		\mathcal{V}_{ei}(k)=\beta^kV_{i}(\varepsilon_i(k)).
	\end{equation}
	Then, the first difference of the $\mathcal{V}_{ei}$ is
	\begin{equation}\label{V_ei}
		\begin{aligned}
			\Delta \mathcal{V}_{ei}%&=\beta^{k+1}V_i(\varepsilon_i(k+1))-\beta^kV_i(\varepsilon_i(k))\\
			%&=\beta^k(\beta V_i(\varepsilon_i(k+1))-V_i(\varepsilon_i(k)))\\
			&=-\beta^kR_i(\varepsilon_i(k),u_i(k_l^i))\leq0.
		\end{aligned}
	\end{equation}
	Overall, $\varepsilon_i(k)$ converges asymptotically, which means that event-triggered formation tracking control is achieved.
\end{proof}
\subsection{Implementation with Online Actor-Critic NNs}
%For the Bellman equation presented in the previous section, 
The optimal controller $u_i^*(k)$ is the unique solution of $(\ref{PIF})$, and $u_i^*(k)$ can be derived by the equation $\dfrac{\partial V_i(\varepsilon_i(k))}{\partial u_i(k)}=0$,
\iffalse
\begin{equation}\label{partial V}
	\dfrac{\partial R_i(\varepsilon_i(k),u_i(k))}{\partial u_i(k)}+\beta\dfrac{\partial V_i(\varepsilon_i(k+1))}{\partial \varepsilon_i(k+1)}\dfrac{\partial\varepsilon_i(k+1)}{\partial u_i(k)}=0.
\end{equation}
Expanding the Eqn. $(\ref{partial V})$, one has $\beta\gamma_i T\dfrac{\partial V_i(\varepsilon_{vi}(k+1))}{\partial \varepsilon_{vi}(k+1)}+2u_i(k)=0$.
%\begin{equation}
%	\begin{aligned}
		%&~~~2u_i(k)+\beta\dfrac{\partial V_i(\varepsilon_{pi}(k+1))}{\partial \varepsilon_{pi}(k+1)}\dfrac{\partial\varepsilon_{pi}(k+1)}{\partial u_i(k)}\\
		%&~~~+\beta\dfrac{\partial V_i(\varepsilon_{vi}(k+1))}{\partial \varepsilon_{vi}(k+1)}\dfrac{\partial\varepsilon_{vi}(k+1)}{\partial u_i(k)}\\
		%&2u_i(k)+\beta\gamma_i T\dfrac{\partial V_i(\varepsilon_{vi}(k+1))}{\partial \varepsilon_{vi}(k+1)}=0.
%	\end{aligned}	
%\end{equation} 
Then, the optimal controller $u_i^*(k)$ is solved as
\fi
\begin{equation}\label{optimal_controller}
	u_i^*(k)=-\dfrac{\beta\gamma_iT}{2}\dfrac{\partial V_i^*(\varepsilon_{vi}(k+1))}{\partial \varepsilon_{vi}(k+1)}.
\end{equation}
where $\gamma_i=\sum_{j\in N_i}w_{ij}+b_i$. Let $\dfrac{\partial V_i^*(\varepsilon_{vi}(k+1))}{\partial \varepsilon_{vi}(k+1)}=\dfrac{2\alpha_p\varepsilon_{pi}(k)}{\beta\gamma_iT}\\+\dfrac{2\alpha_v\varepsilon_{vi}(k)}{\beta\gamma_iT}+\dfrac{V_i^0(\varepsilon_i)}{\beta\gamma_iT}$, where constants $\alpha_p>0$, $\alpha_v>0$, and $V_i^0(\varepsilon_i)=-2\alpha_p\varepsilon_{pi}(k)-2\alpha_v\varepsilon_{vi}(k)+\beta\gamma_iT\dfrac{\partial V_i^*(\varepsilon_{vi}(k+1))}{\partial \varepsilon_{vi}(k+1)}$. Then one has
\begin{equation}\label{controller_split}
	u_i^*(k)=-\alpha\varepsilon_{i}(k)-\dfrac{V_i^0(\varepsilon_i)}{2},
\end{equation}
where $\alpha=[\alpha_p,\alpha_v]\otimes I_m$. And in order to implement the approximation of the optimal controller and the optimized PIF, actor-critic NNs are introduced.
\begin{figure}[!ht]
	\centering
	\includegraphics[width=\columnwidth]{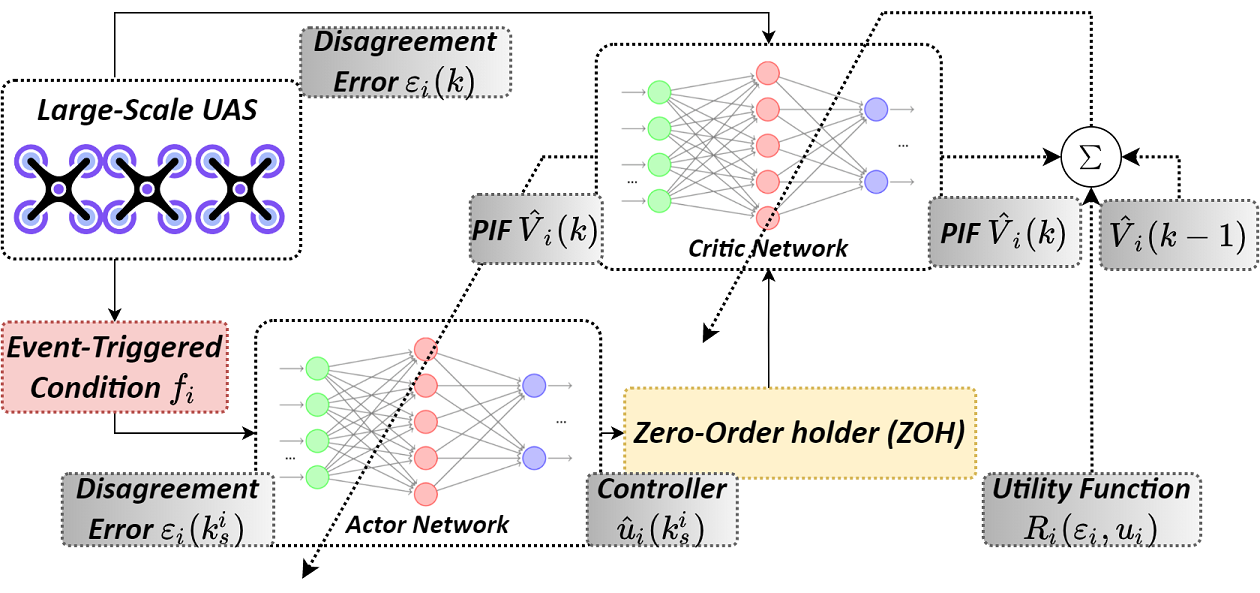}
	\caption{A framework for implementation with event-triggered online actor-critic NNs for formation tracking control. \label{Fig8}}
	\vspace{-2em} 
\end{figure}
\subsubsection{Critic NN Design}
The critic NN, which is utilized to approximate the the optimized PIF,
\iffalse
\begin{equation}\label{critic NN}
	\begin{aligned}
		\dfrac{\partial V_i^*(\varepsilon_{vi}(k+1))}{\partial\varepsilon_{vi}(k+1)}&=\dfrac{2\alpha\varepsilon_{i}(k)}{\beta\gamma_iT}+\dfrac{1}{\beta\gamma_iT}(W_{ci}^{*\mathrm{T}}(k)\psi_i(\varepsilon_i)+\epsilon_{ci,k}),
	\end{aligned}	
\end{equation}
where $W_{i}^*\in\mathbb{R}^{s\times m}$ is the ideal weight of the critic NN with  $s$ neurons, $\psi_i(\varepsilon_i)\in\mathbb{R}^{s}$ is the basis function, and $\epsilon_{ci,k}\in\mathbb{R}^m$ is the approximation error. Since $W_{ci}^*\in\mathbb{R}^{s\times m}$ is unknown at the initial moment, the Eqn. $(\ref{critic NN})$ can be written as
\fi
\begin{equation}\label{PIF NN}
	\begin{aligned}
		\dfrac{\partial \hat{V}_i(\varepsilon_{vi}(k+1))}{\partial\varepsilon_{vi}(k+1)}&=\dfrac{2\alpha\varepsilon_{i}(k)}{\beta\gamma_iT}+\dfrac{1}{\beta\gamma_iT}\hat{W}_{ci}^\mathrm{T}(k)\psi_i(\varepsilon_i),
	\end{aligned}	
\end{equation}
where $\hat{W}_{ci}\in\mathbb{R}^{s\times m}$ is the estimation of $W_{i}^{*}$, which is trained by the critic NN update law
\begin{equation}\label{critic_NN_update}
	\begin{aligned}
		\hat{W}_{ci}(k+1)&=\hat{W}_{ci}(k)-\mu_{ci}T\psi_i(\varepsilon_i)\psi_i^\mathrm{T}(\varepsilon_i)\hat{W}_{ci}(k),
	\end{aligned}
\end{equation}
where $\mu_{ci}$ is the constant gain for the critic NN. %and the weight update law $(\ref{critic_NN_update})$ is obtained from the system stability analysis (see Theorem \ref{theorem2} for details).
\subsubsection{Actor NN Design}The actor NN, which is utilized to approximate the formation tracking controller $(\ref{controller_split})$,
\iffalse
\begin{equation}\label{actor NN}
	\begin{aligned}
		u_i^*(k)=-\alpha\varepsilon_{i}(k)-\dfrac{1}{2}(W_{ai}^{*\mathrm{T}}(k)\psi_i(\varepsilon_i)+\epsilon_{ai,k}),
	\end{aligned}	
\end{equation}
where %$W_{i}^*\in\mathbb{R}^{s\times m}$ is the ideal weight of the actor NN with  $s$ neurons, and 
$\epsilon_{ai,k}\in\mathbb{R}^m$ is the approximation error, %Since $W_{i}^*\in\mathbb{R}^{s\times m}$ is unknown at the initial moment, 
and the Eqn. $(\ref{actor NN})$ can be written as
\fi
\begin{equation}\label{controller NN}
	u_i^*(k)=\begin{cases}
		-\alpha\varepsilon_{i}(k)-\dfrac{1}{2}\hat{W}_{ai}^{\mathrm{T}}(k)\psi_i(\varepsilon_i(k)),~ k\in\{k_l^i\}_{l\in\mathbb{Z}^+},\\
		u_i^*(k_l^i),\qquad\qquad\qquad\qquad\quad~ k\in[k_l^i,k_{l+1}^i).
	\end{cases}	
\end{equation}
where $\hat{W}_{ai}\in\mathbb{R}^{s\times m}$ is the estimation of $W_{ai}^{*}$, which is trained by the actor NN update law
\begin{equation}\label{ET_update_law}
	\hat{W}_{ai}(k+1)=\begin{cases}
		\hat{W}_{ai}(k)-\mu_{ai}T\psi_i(\varepsilon_i)\psi_i^\mathrm{T}(\varepsilon_i)\\
		\quad\times\left(  \hat{W}_{ai}(k)-\hat{W}_{ci}(k)\right),~ k\in\{k_l^i\}_{l\in\mathbb{Z}^+},\\
		\hat{W}_{ai}(k),\qquad\quad\qquad\qquad k\in[k_l^i,k_{l+1}^i).
	\end{cases}
\end{equation}
where $\mu_{ai}$ is the constant gain for actor NN.
\iffalse
\begin{remark}
	In \cite{}, both actor NN and critic NN update the weights continuously or periodically, which can be computationally intensive for complex NNs and consume lots of computing resources. Whereas, on the contrary, from the equation $(\ref{ET_update_law})$ proposed in this paper, the actor NN weight $\hat{W}_{ai}$ is updated only at the triggering instants, which consumes lower computational resources during the flight.
\end{remark}
\fi
\begin{definition}
	A system is said to be ultimatelly uniformly bounded (UUB) if there exist positive constants $\delta$ and $\varrho$, $\forall$$x(k_0)<\delta$, $\exists$$T(\varrho,\delta)\geq0$ such that
	\begin{equation*}
		\forall k>k_0+T, \qquad||x(k)||<\varrho.
	\end{equation*} 	
\end{definition}
\begin{theorem}\label{theorem2}
	For any UAV $a_i$ in the large-scale UAS $(\ref{multi-UAV system})$ with bounded initial states, the optimized event-triggered formation tracking controller $(\ref{controller NN})$ is utilized, where the iterative update laws of the actor-critic NNs are  $(\ref{critic_NN_update})$ and $(\ref{ET_update_law})$, and the constant gains involved satisfy the conditions,
	\begin{equation}\label{conditions}
		\mu_{ci}>\mu_{ai}>0,\quad 0<T<\dfrac{\mu_{ci}-\mu_{ai}}{\mu_{ci}^2\lambda_{\max}^{\psi_i\psi_i^\mathrm{T}}},
	\end{equation}
	where $\lambda_{\max}^{\psi_i\psi_i^\mathrm{T}}$ is the maximum eigenvalues of $\psi_i(k)\psi_i^\mathrm{T}(k)$. Then, the disagreement error $\varepsilon_i(k)$, and the estimation errors $\tilde{W}_{ai}(k)$, $\tilde{W}_{ci}(k)$ of the actor-critic NN weights are UUB, where $\tilde{W}_{ai}(k)=\hat{W}_{ai}(k)-W_{ai}^*$ and $\tilde{W}_{ci}(k)=\hat{W}_{ci}(k)-W_{ci}^*$.
\end{theorem}
\begin{figure*}[hb]% h here; t top; b bottom; p page.
	\vspace{-0.5em}
	\centering
	\includegraphics[width=\linewidth,scale=1.00]{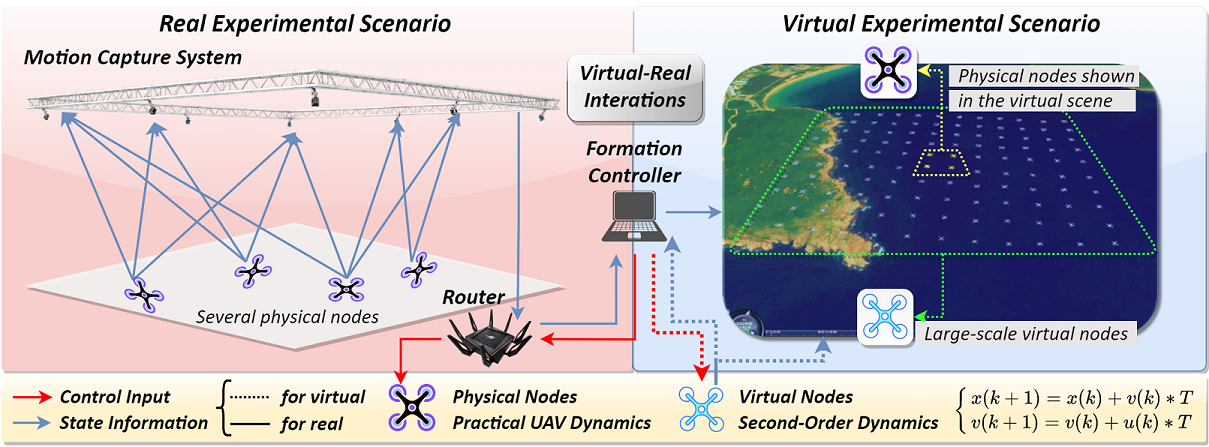}
	\caption{Mixed reality experimental platform framework. To implement the large-scale UAV switching formation control, several physical UAV nodes and large-scale virtual UAV nodes are used to achieve mixed reality interactions.}% The formation controller processes the state information of the physical and virtual nodes to calculate the control inputs, then transmits them to the corresponding nodes. At the same time, the state information of the physical nodes can also be transmitted to the virtual experimental scenario and displayed together with the virtual nodes.}	
	\label{Fig2}
\end{figure*}
\begin{proof}
	The stability analysis process will be divided into two parts according to the trigger instants $k\in\{k_l^i\}_{l\in\mathbb{Z}^+}$ and the trigger instant intervals $k\in[k_l^i,k_{l+1}^i)$.
	
	$\mathbf{(a)}$ For $\forall k\in[k_l^i,k_{l+1}^i)$, the triggering condition $f_i(k)\leq0$ is satisfied, then design the Lyapunov function as follows,
	\begin{equation}
		\begin{aligned}
			\mathcal{V}(k)%&=\mathcal{V}_{e}(k)+\mathcal{V}_{la}(k)+\mathcal{V}_{lc}(k)\\
			&=\sum_{i=1}^{N}\varepsilon_i^\mathrm{T}(k)\varepsilon_i(k)+\dfrac{1}{2}\sum_{i=1}^{N}\mathrm{Tr}\left\lbrace\tilde{W}_{ai}^\mathrm{T}(k) \tilde{W}_{ai}(k)\right\rbrace\\ &~~~+\dfrac{1}{2}\sum_{i=1}^{N}\mathrm{Tr}\left\lbrace\tilde{W}_{ci}^\mathrm{T}(k) \tilde{W}_{ci}(k)\right\rbrace>0.
		\end{aligned}	
	\end{equation}
	For the first difference of $\mathcal{V}(k)$, 
	\begin{equation}
		\begin{aligned}
			\Delta\mathcal{V}(k)&\leq-\sum_{i=1}^N\dfrac{\mu_{ai}T}{2}\mathrm{Tr}\left(\tilde{W}_{ai}^\mathrm{T}(k)\psi_i\psi_i^\mathrm{T}\tilde{W}_{ai}(k)\right)\\
			&~~~-\sum_{i=1}^N\dfrac{(\mu_{ci}-\mu_{ai})T}{2}\mathrm{Tr}\left(\tilde{W}_{ci}^\mathrm{T}(k)\psi_i\psi_i^\mathrm{T}\tilde{W}_{ci}(k)\right)\\
			&~~~+\sum_{i=1}^N\dfrac{\mu_{ci}T}{2}\mathrm{Tr}\left(W_{ci}^*\psi_i\psi_i^\mathrm{T}{W}_{ci}^*\right),
		\end{aligned}
	\end{equation}
	where $\sum_{i=1}^N\dfrac{\mu_{ci}T}{2}\mathrm{Tr}\left(W_{ci}^*\psi_i\psi_i^\mathrm{T}W_{ci}^*\right)\leq\tau$, $\tau>0$. and $\varsigma=\min_{i\in[1,N]}\left\lbrace\mu_{ai}T\lambda_{\max}^{\psi_i\psi_i^\mathrm{T}},(\mu_{ci}-\mu_{ai})T\lambda_{\max}^{\psi_i\psi_i^\mathrm{T}}\right\rbrace$. Then, one has $\Delta\mathcal{V}(k)%&\leq-\dfrac{\varsigma}{2}\sum_{i=1}^N\left\lbrace\mathrm{Tr}\left(\tilde{W}_{ai}^\mathrm{T}(k)\tilde{W}_{ai}(k)\right)\right.\\
			%&\left.~~~+\mathrm{Tr}\left(\tilde{W}_{ci}^\mathrm{T}(k)\tilde{W}_{ci}(k)\right)\right\rbrace+\tau\\
			\leq-\varsigma\mathcal{V}_l(k)+\tau$.
	%where $\mathcal{V}_l(k)=\mathcal{V}_{la}(k)+\mathcal{V}_{lc}(k)$. 
	Then, with $k$ iterations, one has
	\begin{equation}
		\begin{aligned}
			\mathcal{V}(k)&\leq(1-\varsigma)^k\mathcal{V}_l(0)+\dfrac{\tau}{\varsigma}\left(1-(1-\varsigma)^k\right)\\
			&~~~-\mathcal{V}_l(0)+\sum_{i=1}^{N}\varepsilon_i^\mathrm{T}(0)\varepsilon_i(0).
		\end{aligned}
	\end{equation}
%	Combining the inequality $(\ref{term1})$, the disagreement error $\varepsilon(k)$ and all the estimation errors $\tilde{W}_{ai}(k)$, $\tilde{W}_{ci}(k)$ of the actor-critic NN weights are ultimately uniformly bounded.
	
	$\mathbf{(b)}$ For $\forall k\in\{k_l^i\}_{l\in\mathbb{Z}^+}$, the controller $u_i^*(k_l^i)$ is updated, then design the Lyapunov function as follows,
	\begin{equation}
		\begin{aligned}
			\mathcal{V}'(k)%&=\mathcal{V}'_e(k)+\mathcal{V}_{la}(k)+\mathcal{V}_{lc}(K)\\
			&=\sum_{i=1}^N\beta^kV_i(\varepsilon_i(k))+\dfrac{1}{2}\sum_{i=1}^{N}\mathrm{Tr}\left\lbrace\tilde{W}_{ai}^\mathrm{T}(k) \tilde{W}_{ai}(k)\right\rbrace\\ &~~~+\dfrac{1}{2}\sum_{i=1}^{N}\mathrm{Tr}\left\lbrace\tilde{W}_{ci}^\mathrm{T}(k) \tilde{W}_{ci}(k)\right\rbrace>0.
		\end{aligned}
	\end{equation}
	Similarly to the previous case, $\mathcal{V}'(k)$ is bounded.
	%\begin{equation}\label{V'(k)}
	%	\begin{aligned}
	%		\mathcal{V}'(k)%&\leq\mathcal{V}_l(k)-\mathcal{V}_l(0)+\mathcal{V}'(0)\\
	%		&\leq(1-\varsigma)^k\mathcal{V}_l(0)+\dfrac{\tau}{\varsigma}\left(1-(1-\varsigma)^k\right)\\
	%		&~~~-\mathcal{V}_l(0)+\sum_{i=1}^NV_i(\varepsilon_i(0)).
	%	\end{aligned}
	%\end{equation}
	For both cases, the disagreement error $\varepsilon(k)$ of the system and all the estimation errors $\tilde{W}_{ai}(k)$, $\tilde{W}_{ci}(k)$ are UUB.
\end{proof}

\iffalse
\begin{table}[h]
\caption{An Example of a Table}
\label{table_example}
\begin{center}
\begin{tabular}{|c||c|}
\hline
One & Two\\
\hline
Three & Four\\
\hline
\end{tabular}
\end{center}
\end{table}
\fi 
   
\section{Mixed Reality Experiment}\label{sec5}
In this section, a mixed reality experimental platform is constructed to verify the proposed optimal assignment and optimal control algorithms for large-scale UAV formation.
\subsection{Mixed Reality Experimental Platform}
In order to verify the effectiveness of the algorithm, the mixed reality large-scale UAV experimental platform shown in Fig. \ref{Fig2} is established, including a real experimental scenario and a virtual experimental scenario. The real scenario contains only several physical UAV nodes limited by the experimental site, and their state information is obtained by the motion capture system -- OptiTrack, and transmitted to the formation controller via a router. In contrast, the virtual scenario can contain hundreds of virtual UAV nodes, whose dynamics model is designed as a second-order model. The formation controller uses the obtained state information of all nodes to calculate the control input of each node and send it back to each node for formation control. And to visualize the formation effect, the state information of physical nodes is also sent to the virtual scene for real-time display.

In this experiment, 4 physical UAVs and 116 virtual UAVs, a total of 120 UAVs, are used to achieve a large-scale formation. The Laplacian matrix used in the system is
\begin{equation*}
	\setlength{\arraycolsep}{1.2pt}
	\renewcommand{\arraystretch}{1.6}
	L=\left[
	\begin{array}{cccccccc}
		4&-1&-1&0&\cdots&0&-1&-1\\
		-1&4&-1&-1&0&\cdots&0&-1\\
		-1&-1&4&-1&-1&0&\cdots&0\\
		\cdots&\ddots&\ddots&\ddots&\ddots&\ddots&\cdots&\cdots\\
		%\cdots&\cdots&\ddots&\ddots&\ddots&\ddots&\ddots&\cdots\\
		-1&-1&0&\cdots&0&-1&-1&4
	\end{array}
	\right]\in\mathbb{R}^{120\times 120}.
\end{equation*}
The diagonal matrix $B=\mathrm{diag}(\mathrm{mod}(N,2))$. And the other parameters are listed in Table \ref{table_1}.%the diameter of UAVs is about $0.15$ $\mathrm{m}$. The expected switching formation changes from a square to a cross, then to a circle and makes a circular motion. Assuming that the virtual leader always remains at the origin, and the desired trajectory differences $\eta(k)$ of all UAVs are generated by the optimal assignment algorithm. 
\begin{table}[!h]\footnotesize
	\renewcommand{\arraystretch}{1.2}
	\caption{Experimental setup}
	\vspace*{10pt}
	\centering
	\label{table_1}
	%\centering
	\begin{tabular}{l l l}
		\hline\hline %\\[-3mm]
		Parameter & Value & Description  \\[0.3ex] \hline
		r & 0.14 & Radius of UAV (m) \\
		s & 60 &  RBF NN nodes \\  
		$\iota_i$ & 1 &Gaussian Basis Function width\\ 
		$\nu_i$& [-3,3] & Gaussian Basis Function center. \\ 
		$\hat{W}_{ai}(0)$ & $\{0.3\}_{60\times2}$ & Initial actor NN weights \\
		$\hat{W}_{ai}(0)$ & $\{0.3\}_{60\times2}$ & Initial critic NN weights\\
		$\mu_{ai}$ & 6 & Actor NN gain\\	
		$\mu_{ci}$ & 8 & Critic NN gain\\
		$\alpha$ & $[6,4]$ & Parameter in controller $u(k)$\\
		$T$ &$0.01$ & Sampling period (s)\\ [0.8ex]
		\hline\hline
	\end{tabular}
\vspace{-1.5em}
\end{table}
%RBF NN contains 60 nodes whose centers are uniformly distributed in the interval $[-3,3]$, and whose widths are all set to $1$. The initial estimated weights and constant gains for actor-critic NN are $\hat{W}_{ai}(0)=\hat{W}_{ci}(0)=\{0.3\}_{60\times2}$, $\mu_{ai}=6$, $\mu_{ci}=8$. And the other parameters are set as $\alpha=[6,4]$, $\kappa=0.7$, $T=0.01$ $\mathrm{s}$.

\subsection{Results Analysis}
As shown in Fig. \ref{Fig3}, the large-scale UAS successfully generates collision-free optimal formation switching trajectories between multiple formations using the optimal assignment algorithm. With the optimized formation tracking controller $(\ref{controller NN})$, the large-scale mixed reality UAV nodes can follow the generated expected trajectories without collision. And the formation tracking error $\xi(k)$ is bounded as shown in Fig. \ref{Fig4}. With the event-triggered mechanism, the controller update frequency per UAV is significantly reduced. As shown in Fig. \ref{Fig5}, the event-triggered algorithm can reduce the update frequency to 27\% of the time-triggered algorithm, which updates the controller continuously or periodically. And the triggering instants sequence for UAVs $1-30$ in the last $1000$ iterations is also illustrated, where the triggering instants for each UAV are aperiodically and on-demand. Fig. \ref{Fig6} illustrates that the actor and critic NN weights are UUB. And the optimal formation tracking control problem is solved.
\begin{figure}[!htbp]
	%\centering
	%\subfigure[Initial square-shaped large-scale formation]{
	%	\includegraphics[width=0.8\columnwidth]{formation1.png}
	%}
	\subfigure[Cross-shaped mixed reality large-scale formation]{
		\includegraphics[width=0.95\columnwidth]{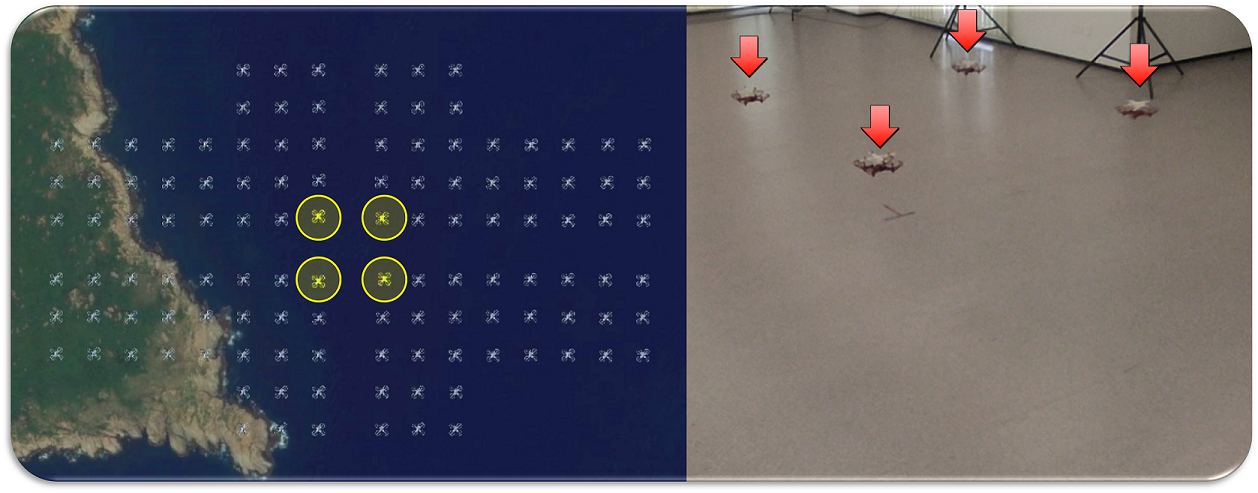}
	}
	\subfigure[Circular mixed reality large-scale formation]{
		\includegraphics[width=0.95\columnwidth]{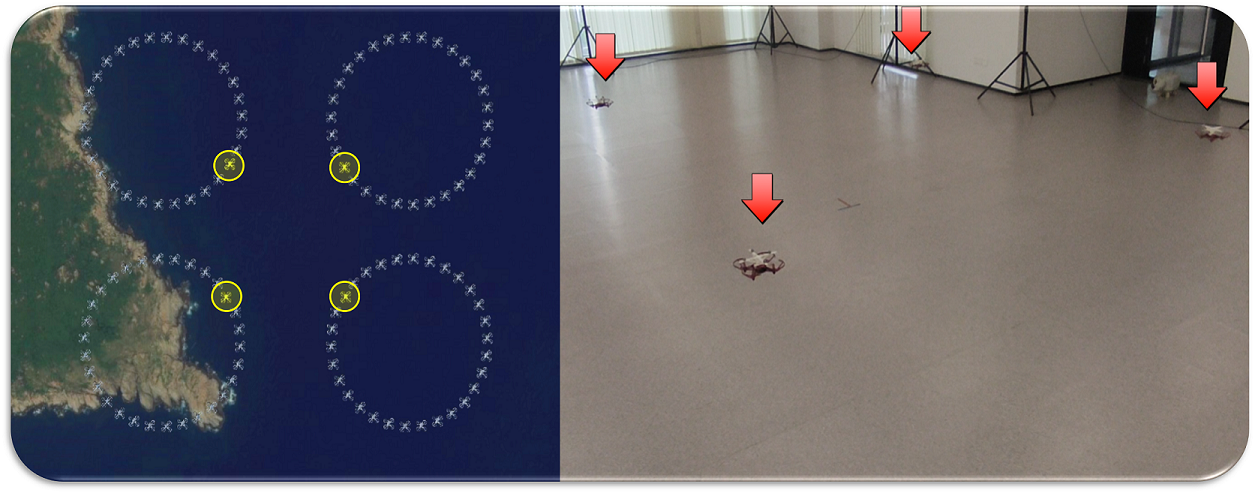}
	}
\vspace{1em} 
	\caption{The large-scale UAV switching formations in the experiment. \label{Fig3}}
	\vspace{-1.5em} 
\end{figure}
\subsection{Summary}
The mixed reality platform built in this study combines the strong engineering applicability of physical experimental platforms with the high flexibility and low cost of virtual simulation platforms. On the one hand, compared with the fully virtual simulations, it can more accurately verify the effectiness of the proposed algorithm for practical UAV dynamics and real-world experimental scenes. On the other hand, compared with the fully physical experiments, the platform has significant advantages in experimental cost and field requirements, which can achieve a heavyweight UAV formation show with lightweight experimental equipment.

\section{Conclusion}\label{sec6}
For the large-scale switching formation tracking problem, this study introduced the optimal formation assignment algorithm and proposed an optimal controller, which broken the main technical limits in decision and control layers of large-scale UASs. Further, compared with the time-triggered controller, the designed event-triggered algorithm reduced the update rate for controller and actor NN to 27\%. More notably, this paper integrated the strong engineering applicability of practical experiments and high flexibility of virtual platforms to practically achieve large-scale switching formation. However, this study ignored the spatial constraints within the practical scene and the dynamics accuracy of the virtual UAV. In future work, autonomous obstacle avoidance in large-scale switching formation will be realized, and the virtual UAV dynamics will be improved in more detail.

\begin{figure}[!htbp]
		\vspace{-1em}
	\centering
	\includegraphics[width=\columnwidth]{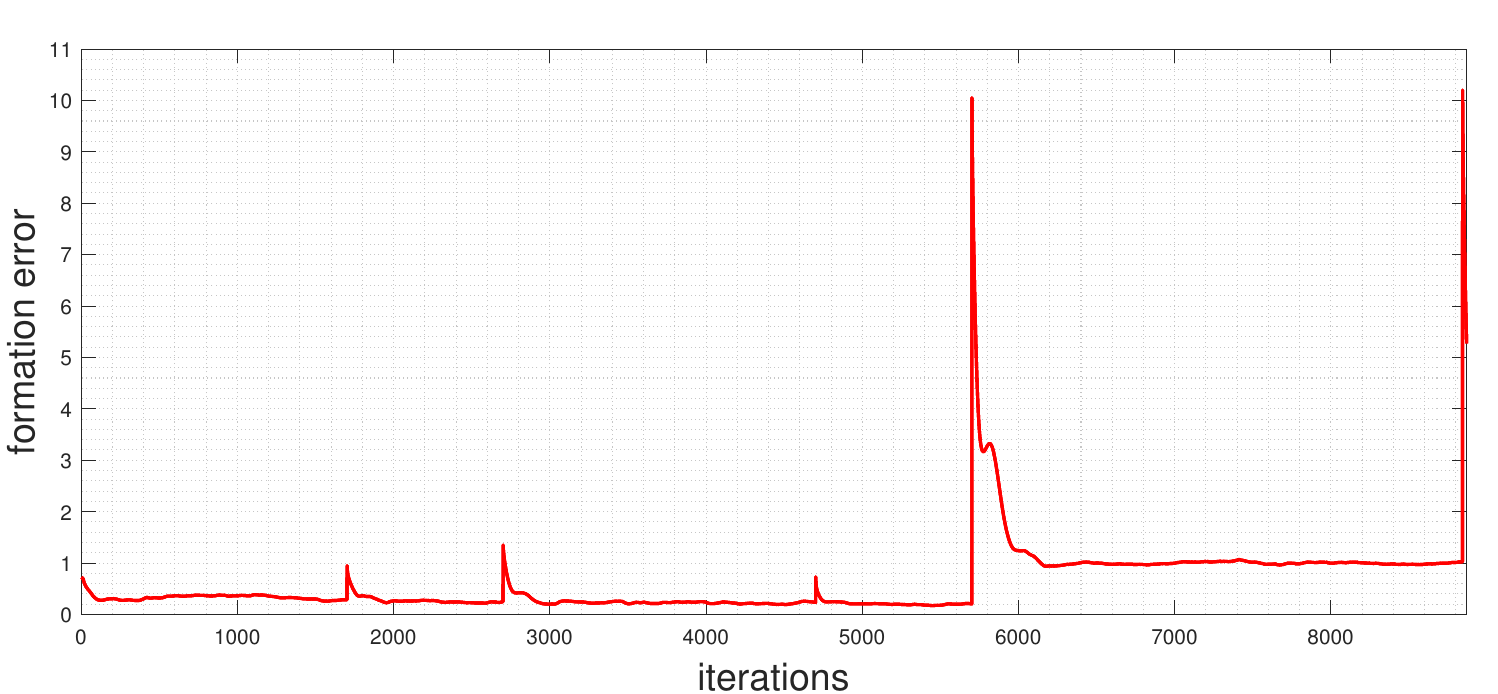}
	\vspace{1em}
	\caption{The norm of formation tracking error $||\xi(k)||$. The moments when the tips appear in the curve are the switch from one formation to another. \label{Fig4}}
\end{figure}
\begin{figure}[!htbp]
		\vspace{-1em}
	\centering
	\includegraphics[width=\columnwidth]{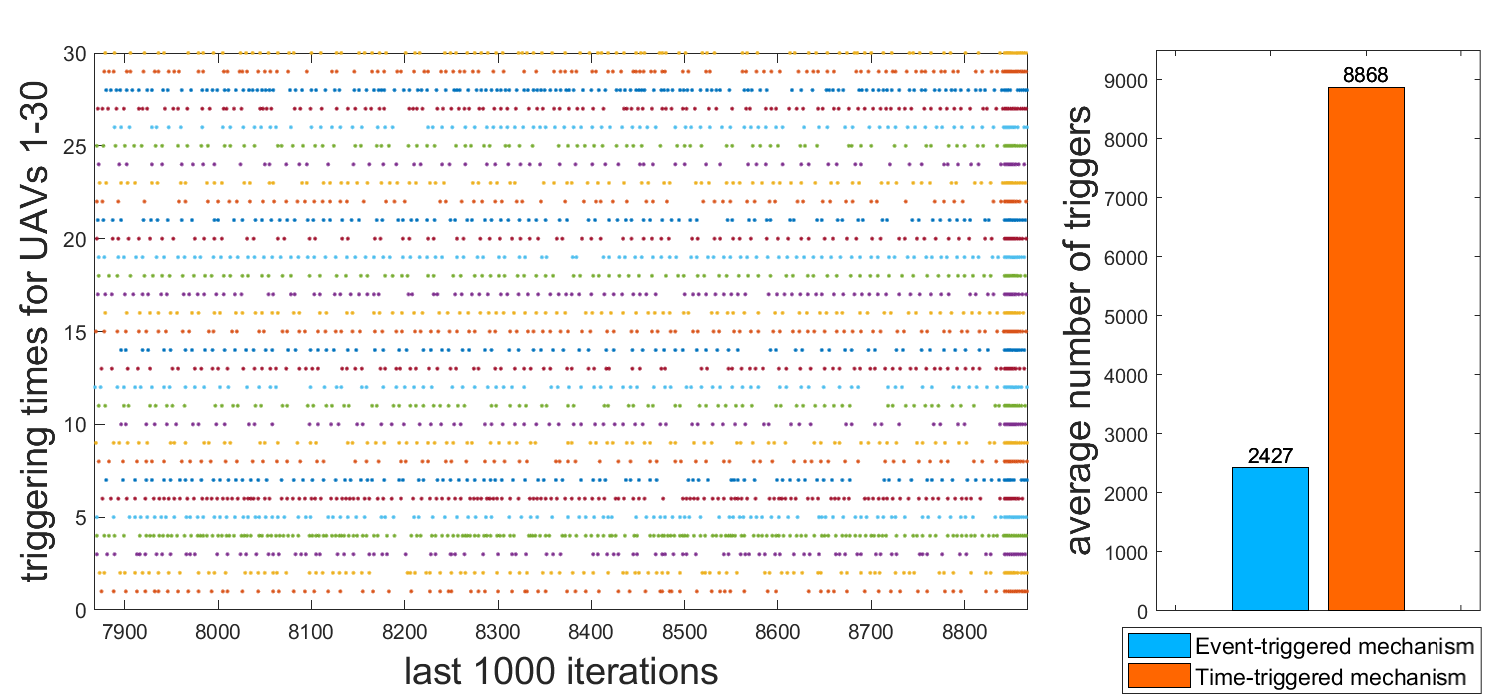}
	\vspace{1em}
	\caption{The triggering instants sequence %in the last 1000 iterations for UAVs numbered $1-30$ 
		and the mean triggers for the event-triggered mechanism compared to the time-triggered mechanism. %over the total iterations. 
		\label{Fig5}}
\end{figure}
\begin{figure}[!htbp]
		\vspace{-1em}
	\centering
	\includegraphics[width=\columnwidth]{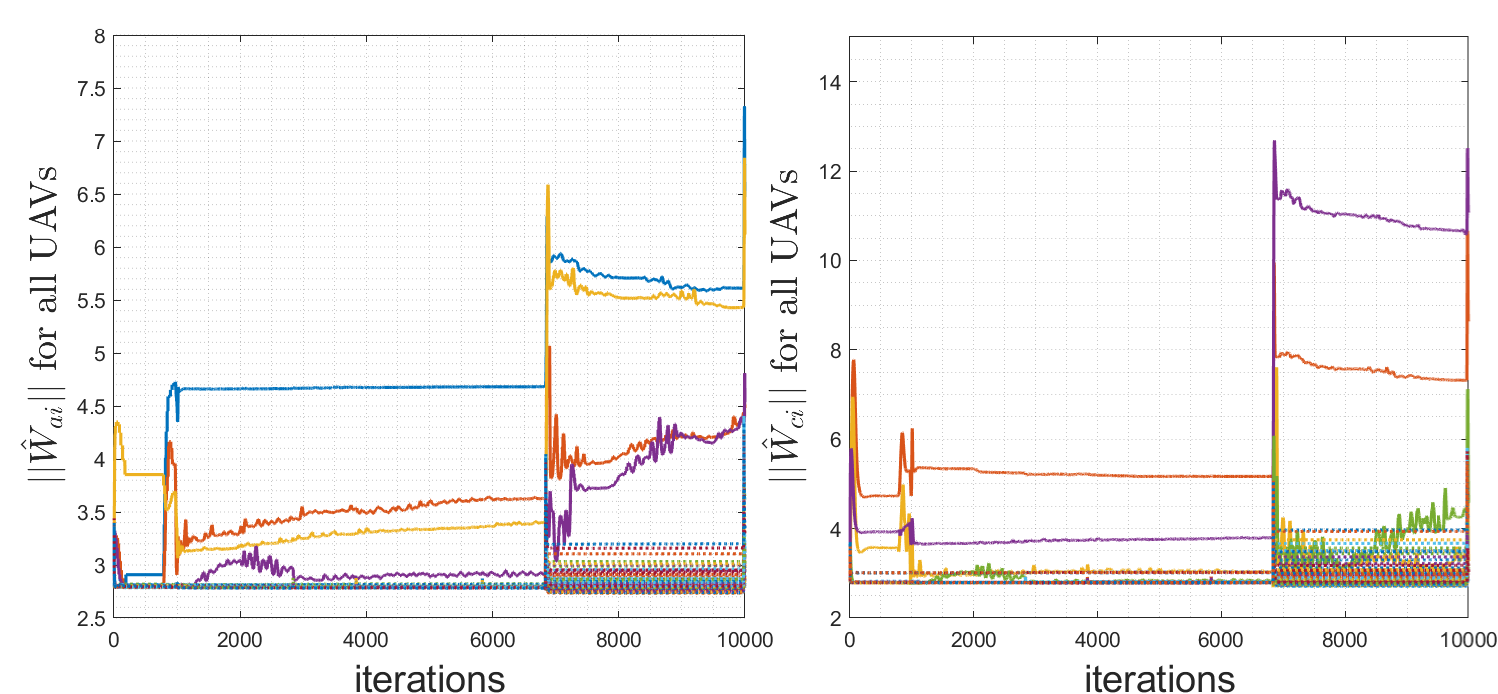}
	\vspace{1em}
	\caption{Actor-critic NN weight norm for all UAVs, where the solid lines represent the physical UAVs and the dotted lines represent the virtual UAVs. \label{Fig6}}
\end{figure}

\addtolength{\textheight}{-6cm}   % This command serves to balance the column lengths

%%%%%%%%%%%%%%%%%%%%%%%%%%%%%%%%%%%%%%%%%%%%%%%%%%%%%%%%%%%%%%%%%%%%%%%%%%%%%%%%

%%%%%%%%%%%%%%%%%%%%%%%%%%%%%%%%%%%%%%%%%%%%%%%%%%%%%%%%%%%%%%%%%%%%%%%%%%%%%%%%

%%%%%%%%%%%%%%%%%%%%%%%%%%%%%%%%%%%%%%%%%%%%%%%%%%%%%%%%%%%%%%%%%%%%%%%%%%%%%%%%


\begin{thebibliography}{99}

\bibitem{c1} H. Zhang, T. Feng, H. Liang, et al, "LQR-Based Optimal Distributed Cooperative Design for Linear Discrete-Time Multiagent Systems," IEEE Trans. Neural Netw. Learn. Syst., vol. 28, no. 3, pp. 599-611, Mar. 2017.

\bibitem{c2} C. O. Aguilar and A. J. Krener, "Numerical Solutions to the Bellman Equation of Optimal Control," J. Optim. Theory Appl., vol. 160, no. 2, pp. 527-552, Feb. 2014.

\bibitem{c3} B. Kiumarsi, F. L. Lewis, H. Modares, et al, "Reinforcement -learning for optimal tracking control of linear discrete-time systems with unknown dynamics," Automatica, vol. 50, no. 4, pp. 1167-1175, Apr. 2014.

\bibitem{c4} H. J. Kappen, "Optimal control theory and the linear Bellman equation," in Bayesian Time Series Models, 1st ed., D. Barber, A. T. Cemgil, and S. Chiappa, Eds. Cambridge University Press, 2011, pp. 363-387. 

\bibitem{c5} R. Kamalapurkar, J. A. Rosenfeld, and W. E. Dixon, "Efficient model-based reinforcement learning for approximate online optimal control," Automatica, vol. 74, pp. 247-258, Dec. 2016.

\bibitem{c6} B. Pang, Z. P. Jiang, and I. Mareels, "Reinforcement learning for adaptive optimal control of continuous-time linear periodic systems," Automatica, vol. 118, p. 109035, Aug. 2020.

\bibitem{c7} A. Perrusquia and W. Yu, "Identification and optimal control of nonlinear systems using recurrent neural networks and reinforcement learning: An overview," Neurocomputing, vol. 438, pp. 145-154, May 2021.

\bibitem{c8} J. Zhang, Z. Wang, and H. Zhang, "Data-Based Optimal Control of Multiagent Systems: A Reinforcement Learning Design Approach," IEEE Trans. Cybern., vol. 49, no. 12, pp. 4441-4449, Dec. 2019.

\bibitem{c81} H. Liu, Q. Meng, F. Peng, et al, "Heterogeneous formation control of multiple UAVs with limited-input leader via reinforcement learning," Neurocomputing, vol. 412, pp. 63-71, Oct. 2020.

\bibitem{c82} W. Zhao, H. Liu and F. L. Lewis, "Robust Formation Control for Cooperative Underactuated Quadrotors via Reinforcement Learning," IEEE Trans. Neural Netw. Learn. Syst., vol. 32, no. 10, pp. 4577-4587, Oct. 2021.

\bibitem{c9} Y. Guo, G. Chen, and T. Zhao, "Learning-based collision-free coordination for a team of uncertain quadrotor UAVs," Aerosp. Sci. Technol., vol. 119, p. 107127, Dec. 2021.

\bibitem{c10} H. Li, Y. Wu, and M. Chen, "Adaptive Fault-Tolerant Tracking Control for Discrete-Time Multiagent Systems via Reinforcement Learning Algorithm," IEEE Trans. Cybern., vol. 51, no. 3, pp. 1163-1174, Mar. 2021.

%\bibitem{c11} T. Li, W. Bai, Q. Liu, et al, "Distributed Fault-Tolerant Containment Control Protocols for the Discrete-Time Multiagent Systems via Reinforcement Learning Method," IEEE Trans. Neural Netw. Learn. Syst., pp. 1-13, Nov. 2021.

\bibitem{c12}M. Long, H. Su, and Z. Zeng, "Model-Free Event-Triggered Consensus Algorithm for Multiagent Systems Using Reinforcement Learning Method," IEEE Trans. Syst. Man Cybern. Syst., vol. 52, no. 8, pp. 5212-5221, Aug. 2022.

\bibitem{c13} X. Guo, W. Yan, and R. Cui, "Event-Triggered Reinforcement Learning-Based Adaptive Tracking Control for Completely Unknown Continuous-Time Nonlinear Systems," IEEE Trans. Cybern., vol. 50, no. 7, pp. 3231-3242, Jul. 2020.

\bibitem{c14} B. Yan, P. Shi, and C. C. Lim, "Robust Formation Control for Nonlinear Heterogeneous Multiagent Systems Based on Adaptive Event-Triggered Strategy," IEEE Trans. Autom. Sci. Eng., pp. 1-13, Aug. 2021.


\bibitem{c14_1} V. Narayanan and S. Jagannathan, "Event-Triggered Distributed Control of Nonlinear Interconnected Systems Using Online Reinforcement Learning With Exploration," IEEE Trans. Cybern., vol. 48, no. 9, pp. 2510-2519, Sep. 2018.

\bibitem{c14_2} X. Yang, H. He, and D. Liu, "Event-Triggered Optimal Neuro-Controller Design With Reinforcement Learning for Unknown Nonlinear Systems," IEEE Trans. Syst. Man Cybern. Syst., vol. 49, no. 9, pp. 1866-1878, Sep. 2019.

\bibitem{c14_3} X. Yang and H. He, "Decentralized Event-Triggered Control for a Class of Nonlinear-Interconnected Systems Using Reinforcement Learning," IEEE Trans. Cybern., vol. 51, no. 2, pp. 635-648, Feb. 2021.

\bibitem{c15} H. Li, Y. Wu, M. Chen, et al, "Adaptive Multigradient Recursive Reinforcement Learning Event-Triggered Tracking Control for Multiagent Systems," IEEE Trans. Neural Netw. Learn. Syst., pp. 1-13, Jul. 2021.

\bibitem{c16} Z. Peng, R. Luo, J. Hu, et al, "Distributed Optimal Tracking Control of Discrete-Time Multiagent Systems via Event-Triggered Reinforcement Learning," IEEE Trans. Circuits Syst. Regul. Pap., pp. 3689-3700, Jun. 2022.

\bibitem{c17} W. Bai, T. Li, Y. Long, et al, "Event-Triggered Multigradient Recursive Reinforcement Learning Tracking Control for Multiagent Systems," IEEE Trans. Neural Netw. Learn. Syst., pp. 1-14, Jul. 2021.

\bibitem{c17_1} J. Lu, Q. Wei, Y. Liu, et al, "Event-Triggered Optimal Parallel Tracking Control for Discrete-Time Nonlinear Systems," IEEE Trans. Syst. Man Cybern. Syst., vol. 52, no. 6, pp. 3772-3784, Jun. 2022.

\bibitem{c17_2} S. Zhao, J. Wang, H. Wang, et al, "Goal representation adaptive critic design for discrete-time uncertain systems subjected to input constraints: The event-triggered case," Neurocomputing, vol. 492, pp. 676-688, Jul. 2022.

\bibitem{c17_3} F. Tang, B. Niu, G. Zong, et al, "Periodic event-triggered adaptive tracking control design for nonlinear discrete-time systems via reinforcement learning," Neural Netw., vol. 154, pp. 43-55, Oct. 2022.

\bibitem{c18} S. Agarwal and S. Akella, "Simultaneous Optimization of Assignments and Goal Formations for Multiple Robots," in IEEE Int. Conf. on Robotics and Automation (ICRA), Brisbane, QLD, pp. 6708-6715, May 2018.


\end{thebibliography}
 \end{document}